\documentclass[sigconf]{acmart}

\usepackage{caption}
\usepackage{subcaption} 
\usepackage{mathtools}
\usepackage{mathbbol}
\usepackage{enumerate}
\usepackage{xcolor}

\usepackage{amsmath}
\usepackage{graphicx}
\usepackage{graphics}
\usepackage{amsthm}
\usepackage{thmtools, thm-restate}
\usepackage{hyperref}
\usepackage{cleveref}

\newtheorem{theorem}{Theorem}[section]

\newtheorem{lemma}[theorem]{Lemma}

\newcommand{\Proba}{\mathbb{P}}
\newcommand{\expect}{\mathbb{E}}

\newcommand{\x}{\mathbf x}

\newcommand{\card}[1]{\left|#1\right|}

\newcommand{\R}{\mathcal{R}}

\newcommand{\N}{\mathbb{N}}

\newcommand{\y}{\mathbf y}


\title[Almost Tight Lower Bound for Plurality Consensus]{An Almost Tight Lower Bound for Plurality Consensus with Undecided State Dynamics in the Population Protocol Model}


\author{Antoine El-Hayek}
\affiliation{%
  \institution{Institute of Science and Technology Austria}
  \city{Klosterneuburg}
  \country{Austria}}
\email{antoine.el-hayek@ist.ac.at}
\orcid{0000-0003-4268-7368}

\author{Robert Els\"asser}
\affiliation{%
  \institution{University of Salzburg}
  \city{Salzburg}
  \country{Austria}
}\email{elsa@cs.plus.ac.at}
\orcid{0000-0002-5766-8103}

\author{Stefan Schmid}
\affiliation{%
 \institution{TU Berlin}
 \city{Berlin}
 \country{Germany}}\email{stefan.schmid@tu-berlin.de}
 \orcid{0000-0002-7798-1711}

\acmYear{2025}\copyrightyear{2025}
\acmConference[PODC '25]{ACM Symposium on Principles of Distributed Computing}{June 16--20, 2025}{Huatulco, Mexico}
\acmBooktitle{ACM Symposium on Principles of Distributed Computing (PODC '25), June 16--20, 2025, Huatulco, Mexico}
\acmDOI{10.1145/3732772.3733505}
\acmISBN{979-8-4007-1885-4/25/06}

\begin{document}

\begin{abstract}
We revisit the majority problem in the population protocol communication model, as first studied by Angluin et al. (Distributed Computing 2008).
We consider a more general version of this problem known as plurality consensus, which has already been studied intensively in the literature.
In this problem, each node in a system of $n$ nodes, has initially one of $k$ different opinions, and they need to agree on the (relative) majority opinion.
In particular,
we consider the important and intensively studied model of Undecided State Dynamics.

Our main contribution is an almost tight lower bound on the stabilization time: we prove
that there exists an initial configuration, even with bias $\Delta = \omega(\sqrt{n\log n})$, where stabilization requires $\Omega(kn\log \frac {\sqrt n} {k \log n})$ interactions, or equivalently, $\Omega(k\log \frac {\sqrt n} {k \log n})$ parallel time for any $k = o\left(\frac {\sqrt n}{\log n}\right)$.
This bound is tight for any $ k \le n^{\frac 1 2 - \epsilon}$, where $\epsilon >0$ can be any small constant, as Amir et al.~(PODC'23) %
gave a $O(k\log n)$ parallel time upper bound for $k = O\left(\frac {\sqrt n} {\log ^2 n}\right)$.
\end{abstract}

\maketitle
\sloppy

\section{Introduction}

Population protocols are a simple and natural computational framework, in which $n$ anonymous nodes (also called agents) communicate and interact with each other to solve a predefined problem in a distributed manner. In the underlying communication model, a scheduler selects in each discrete time step two nodes for interaction. The selected nodes exchange their current states and each of them changes its own state according to the transition function defined by the population protocol. The set of nodes is expected to eventually stabilize in a final configuration that is given by the problem definition. In the plurality consensus problem, as considered in this paper, every node has at the beginning an opinion assigned from a set of $k$ different opinions, and in the final configuration all nodes agree on one of the initial opinions. 

In the model introduced by Angluin et al.~\cite{DBLP:journals/dc/AngluinADFP06} the population is represented by the nodes of a graph and a scheduler can only choose nodes connected by an edge for interaction. The usual complexity measures in which one is interested are the cardinality of the state space of the transition function and the time needed for the population to stabilize. The time is defined as the number of interactions until a stable configuration is reached. As in previous work, we are also interested in the so-called parallel time, which corresponds to the number of interactions divided by the population size $n$. 

Population protocols have various applications. In their original paper, Angulin et al.~\cite{DBLP:journals/dc/AngluinADFP06} motivated population protocols in the context of sensor networks where nodes perform simple computations. Other motivating examples are processes in chemical reaction networks~\cite{SCWB08}. %
Population protocols can also be implemented at the level of DNA molecules, as shown in~\cite{CDS+13}. Furthermore, 
Cardelli and Csiksz-Nagy~\cite{CCN12} considered similarities between biochemical regulatory processes in living cells 
and population protocols. Population protocols highly influenced the way certain aspects of distributed computing evolved in recent years, for which in 2020 the original paper by Angluin et al.~\cite{DBLP:journals/dc/AngluinADFP06} has been awarded the Edsger W. Dijkstra Prize in Distributed Computing.

In the vast majority of the related papers, as in our paper as well, it is assumed that the graph modeling the population is a clique and a random scheduler is in place. That is, in each discrete time step two nodes are selected uniformly at random for interaction. One of the most prominent protocols for the plurality consensus problem is the Undecided State Dynamics, on which we focus in this paper. In its original, unconditional version, this protocol is simple and only uses $k+1$ different states (see the problem definition below for details). This protocol and variants of it have extensively been analyzed in different communication models. Although very recently an upper bound of $O(k \log n)$ on the parallel time has been derived~\cite{DBLP:conf/podc/AmirABBHKL23}, the question of whether this bound is tight is still open.

\subsection{Model and problem}

A population protocol and its time performance can be formalized as follows
(cf.~\cite{BEFKKR18}). 
Let $V$ denote the set of agents in the population, and let $n = |V|$.
Let $\Sigma$ be the set of states of the protocol, whose cardinality may grow
with~$n$. 
Two interacting nodes change their states
according to a deterministic function 
$f: \Sigma^2 \mapsto \Sigma^2$. 
That is, $f(q', q'') = (r', r'')$ describes the following transition:
if a node in state $q'$ interacts with a node in state $q''$, 
then the first node changes its state to $r'$ and the second node to $r''$.
A population protocol is also assigned an {\em output function\/} $\gamma: \Sigma \rightarrow \Gamma$, which maps every state to an output value. 
The set $\Gamma$ may be the same as $\Sigma$.

In this paper we analyze plurality consensus (aka the $k$-majority problem):
Each agent~$i$ has as an input an initial opinion $s_i \in [k]$, and the goal of the agents is to decide which opinion was the (relative) majority at the start of the computation.
The transition function we consider is given by the (unconditional) Undecided State Dynamics: $\Sigma$ consists of $k+1$ states, the original $k$ opinions, and an extra one, $\bot$, that represents the state of being undecided.
Then $f(s_1, s_2) = (\bot, \bot)$ if $s_1 \neq s_2$, and $s_1, s_2 \in [k]$, $f(s, \bot) = (s, s)$ for any $s \in [k]$.
Otherwise, $f$ is just the identity function.
In other words, when two agents with different opinions meet, they both become undecided, but when a decided agent meets an undecided one, the latter takes on the opinion of the former. Note that in the case of the Undecided State Dynamics the set $\Gamma = \Sigma$ and the output function $\gamma$ is the identity. We assume that the interaction graph is a clique, and at each time step, two nodes are selected for interaction, which are chosen uniformly at random (without replacement), independently of the other time steps. 
We ask how long it takes for the system to stabilize, in the particular case where the majority opinion starts with an initial additional bias of $\Omega (\sqrt{n\log n})$. Note that such a bias is probably needed in order to guarantee that w.h.p. the opinion with the largest (relative) initial support wins (cf.~\cite{DBLP:conf/podc/AmirABBHKL23,10.1137/1.9781611977073.135}).

\subsection{Related work}

Computations over dynamic networks (such as sensor networks) have already been studied intensively in the literature for various models~\cite{afek2013asynchrony,kuhn2011dynamic,michail2020distributed}. 
Probably the most studied problems in the framework of population protocols are leader election and (exact) majority. For two recent surveys, which focus on these problems, we refer the reader to~\cite{DBLP:journals/sigact/AlistarhG18} and~\cite{DBLP:journals/eatcs/ElsasserR18}. As in this paper we analyze plurality consensus, we only present the most relevant results on the related majority problem and adapt the description of~\cite{DBLP:journals/eatcs/ElsasserR18} to outline the related work on majority in population protocols. In this problem, at the beginning every node is in one of two states (called e.g. $A$ and $B$).
In exact majority, the final opinion has to be the one with the largest initial support, even if at the beginning the difference between the support of the two was just $1$. In approximate majority this requirement is weaker: the initial majority should only win with high probability if there is a sufficiently large initial bias between the two opinions (usually this bias is of order $\Omega(\sqrt{n \log n})$).

For the majority problem (as well as leader election), there are a number of results, which present lower bounds on the number of states 
under certain time requirements, 
or bound the (stabilization) time under specific assumptions 
w.r.t. the number of states. 
Furthermore, 
several
majority and leader election algorithms have been derived, which upper bound the number of states 
as well as the stabilization time. 

Two early papers by Draief and Vojnovi\'c~\cite{DBLP:conf/infocom/DraiefV10} and Mertzios et al.~\cite{Mertzios-etal-ICALP2014} 
consider population protocols for exact majority. They studied 
(almost) the same four-state protocol, 
which has a polynomial stabilization time (with high probability\footnote{With high probability or w.h.p. means with probability at least $1-n^{-\Omega(1)}$, where $n$ is the number of agents.} and in expectation) on any graph. 

In an early paper, Angluin et al.~\cite{AngluinAE2008fast} 
presented population protocols with a constant number of states 
for several different functions. 
The protocols they propose are only correct with high probability and they assume 
that a designated leader is available from the start of the computation, which
synchronizes the nodes. 
Their exact majority protocol has a w.h.p. stabilization time of
$O(\log^2 n)$, and it alternates between so-called cancellation and duplication phases, an idea used in many subsequent papers.

Note that any protocol for exact majority which uses a constant number of states (as the four-state protocol described above) is in general slow. 
However, if the initial imbalance between the support of $A$ and $B$ is large, then the four-state protocol stabilizes fast. 
In order to increase the initial imbalance, Alistarh et al.~\cite{DBLP:conf/podc/AlistarhGV15} 
multiplied the opinion on each node by some (large) integer. 
Based on this
idea, Alistarh et al.~\cite{DBLP:conf/soda/AlistarhAEGR17} achieved a stabilization time of $O(\log^3 n)$, w.h.p.~and
in expectation, by utilizing $O(\log^2 n)$ states.

Bilke et al.~\cite{DBLP:conf/podc/BilkeCER17} extended the cancellation-duplication
framework from~\cite{AngluinAE2008fast} to the leaderless case, provided that the agents
have enough states to store the number of interactions they performed so far. The stabilization time of their majority protocol 
is $O(\log^2 n)$ w.h.p.~and in expectation, and it utilizes $O(\log^2 n)$
states.

On the lower bound side, Alistarh et al.~\cite{DBLP:conf/soda/AlistarhAG18} proved
that any exact majority
protocol with expected stabilization time $O(n^{1-\epsilon})$ ($\epsilon$ can be any
positive constant), which satisfies two natural properties called {\em monotonicity} 
and {\em output dominance}, 
requires $\Omega(\log n)$ states. They also presented an algorithm
with $\Theta(\log n)$ states and $O(\log^2 n)$ stabilization time (w.h.p.~and
in expectation). 
Monotonic protocols have the property that their running time does not increase
if they are run on a smaller number of agents.
Output dominance means that
``if the positive counts of states in a stable configuration are changed, then
the protocol will stabilize to the same output" (cf.\cite{DBLP:journals/eatcs/ElsasserR18}). The (w.h.p.~and expected) stabilization time has subsequently been improved to $O(\log^{5/3} n)$ in~\cite{BEFKKR18}, to $O(\log^{3/2} n)$ in~\cite{NunKKP20}, and finally to 
$O(\log n)$ in~\cite{DotyEGSUS21}, by keeping the number of states at $O(\log n)$.

All the results above were derived for \emph{stabilization} time, and the lower bounds do not hold for the so-called \emph{convergence} time.
The convergence time is the time required by the protocol to reach a configuration with the correct output property; however, the system may leave this configuration with a small probability. In contrast, if the system stabilizes, then the output of the system does not change anymore.\footnote{See e.g.~\cite{BEFK+18} for details. In the Undecided State Dynamics, convergence and stabilization are equivalent.}
Kosowski and Uzna{\'n}ski~\cite{DBLP:conf/podc/KosowskiU18} and 
Berenbrink et al.~\cite{BEFK+18}
derived algorithms with polylogarithmic \emph{convergence} time, which use $o(\log n)$ states.  
As outlined in~\cite{BEFKKR18}, in~\cite{DBLP:conf/podc/KosowskiU18} the authors presented a programming framework that leads to protocols
which require only $O(1)$ states and 
converge in polylogarithmic time (in expectation), but they are only correct w.h.p. These protocols can be changed so that they are
always correct by either allowing $O(\log\log n)$ states, while the convergence time still remains polylogarithmic, or by allowing 
$O(n^\epsilon)$ convergence time, while keeping the number of states constant.
In~\cite{BEFK+18} the authors presented an always correct
protocol with a w.h.p. convergence time of
$O(\log^2 n/{\log s})$ and $\Theta(s + \log\log n)$ states,
and an always correct
protocol with w.h.p. stabilization time of
$O(\log^2 n/{\log s})$ and $O(s \cdot \log n / {\log s})$ states,
where $s\in [2,n]$.

One research direction in plurality consensus focuses on the state complexity (regardless of the time complexity) of protocols, which are required to \emph{always} determine the plurality opinion.
While clearly at least $k$ states are required to encode $k$ opinions,~\cite{DBLP:conf/ciac/NataleR19} shows that always correct plurality consensus needs even $\Omega(k^2)$ states.
The protocol of~\cite{DBLP:conf/opodis/GasieniecHMSS16} utilizes $O(k^{11})$ states, which can be improved to $O(k^6)$ provided that a total ordering among the opinions exists.
Clearly, the lower bound of $\Omega(k^2)$ only holds, if it is required that the correct plurality opinion is determined with probability $1$.
If such strong guarantees are not required, then the number of states  can be much smaller. In~\cite{10.1137/1.9781611977073.135} a synchronized variant of the Undecided State Dynamics has been presented that reaches consensus w.h.p.~in $O(\log^2 n)$ parallel time using $O(k \log n)$ states.
However, this protocol solves \emph{approximate} plurality consensus, i.e., if the 
initial bias is $\Omega(\sqrt{n \log n})$, then w.h.p.~the opinion with the initially largest support wins, otherwise a so-called significant opinion wins w.h.p. In the case $k=2$, the unconditional Undecided State Dynamics has a w.h.p. and expected stabilization time of $O(\log n)$~\cite{CGGNPS18}. Recently, Amir et al.~\cite{DBLP:conf/podc/AmirABBHKL23} analyzed the unconditional Undecided State Dynamics for the plurality consensus problem and showed that their protocol stabilizes w.h.p.~within $O(k \log n)$ parallel time for any initial configuration as long as $k = O(\sqrt{n}/\log^2 n)$. The question of whether this bound on the stabilization time is tight is still open.

We should note that the unconditional Undecided State Dynamics has extensively been analyzed in the Gossip communication model. In this model (which can be seen as a synchronous variant of the population protocol model) in each discrete time step, every node randomly chooses another node for interaction to perform a state transition. Becchetti et al. defined the concept of \emph{monochromatic distance} $\mbox{md}(c)$ of a configuration $c$ and showed that in this model the time needed to reach a final configuration is $O(\mbox{md}(c) \log n)$ w.h.p., where $c$ is the initial configuration in the population~\cite{BCNPS15}. They also derived a lower bound that is asymptotically tight up to a $\log n$ factor.

As described by Amir et al.~\cite{DBLP:conf/podc/AmirABBHKL23}, the differences in how nodes are scheduled for interaction in the population protocol model and the Gossip model, respectively, cause the Undecided State Dynamics to ``exhibit significant qualitative differences when run in either setting, even in the case when $k=2$''. One of the reasons for these differences is the fact that while in the Gossip model in each step a node may change its opinion only once, and each node is selected for interaction, in the population protocol model a node may change its opinion up to $\Omega(\log n)$ many times in $n$ consecutive interactions (which corresponds to one parallel round) while a constant fraction of nodes is not even selected for interaction. 
Thus, there are so far no general analysis techniques that allow us to transfer results from one model to the other one (cf.~\cite{CGGNPS18}).

\subsection{Contribution}

Our main contribution is an almost tight lower bound on the stabilization time of Undecided State Dynamics for plurality consensus in population protocols. 
Specifically, we show that the time needed to stabilize is $\Omega\left(kn \log \frac {\sqrt n} {k\log n}\right)$ interactions, or in $\Omega\left(k \log \frac {\sqrt n} {k\log n}\right)$ parallel time\footnote{Recall that the parallel time is equal to the number of interactions divided by $n$.}, in the case where $k = o(\frac {\sqrt n} { \log n})$. In our proofs we assume that our initial configuration has bias at most $O(\sqrt{n}/(k \log n))^{1/4} \sqrt{n \log n}$. Interestingly, this includes even an initial bias of $\omega(\sqrt{n \log n})$. 
This bound is tight for any $ k \le n^{\frac 1 2 - \epsilon}$, where $\epsilon >0$ can be any small constant, as Amir, Aspnes, Berenbrink, Biermeier, Hahn, Kaaser, Lazarsfeld~\cite{DBLP:conf/podc/AmirABBHKL23} gave a $O(k\log n)$ parallel time upper bound. 
For larger values of $k$, as any initial configuration that is valid for $k_0$ is also valid for $k\ge k_0$, we can simply plug in $k_0 = \frac {\sqrt n} {\log n \log \log n}$ to get a $\Omega\left(\frac {\sqrt n \log \log \log n} {\log n \log \log n}\right)$ parallel time lower bound. 

Our analysis is based on a precise characterization of the number of undecided nodes over time, using drift analysis and random walks. Our technical approach and its novelty will be discussed in more details in the next section. 

\subsection{Notation and organization}

Throughout the paper, $\x$ represents a configuration the system can be in: $\x = (x_1, \dots, x_k, u)$, where $x_i$ is the number of agents with opinion $i$ for all $i$, and $u$ is the number of undecided agents.
We denote by $\x(t)$ the configuration of the system after the $t$-th interaction, and $\x(0)$ is the initial configuration.
$x_i(t)$ is the number of agents with opinion $i$ after interaction~$t$. Similarly, $u(t)$ is the number of undecided nodes after $t$ interactions. 
Accordingly, $x_i(0)$ is the initial number of nodes with opinion $i$ (clearly, $u(0)=0$). 
We assume that the opinions are initially ordered from most common opinion to least common, i.e, $x_1(0) \ge x_2(0) \ge \dots \ge x_k(0)$.

The remainder of this paper is organized as follows. We first provide a technical overview and an empirical motivation in Section \ref{sec:experiments}. Our lower bound analysis is presented in Section \ref{sec:lowerbound}. We conclude and discuss open research questions in Section \ref{sec:conclusion}. Some technical details are deferred to the Appendix.

\section{Motivation and Technical Overview}\label{sec:experiments}

Before we delve into the technical details of our lower bound proof, let us provide some intuition about how agent opinions propagate across
the distributed system, and how majority and minority opinions evolve accordingly. Based on this intuition and the observed challenges, 
we will then give an overview of our analytical approach and its key ideas.

\begin{figure*}[t]
    \centering
    \includegraphics[width=0.49\linewidth]{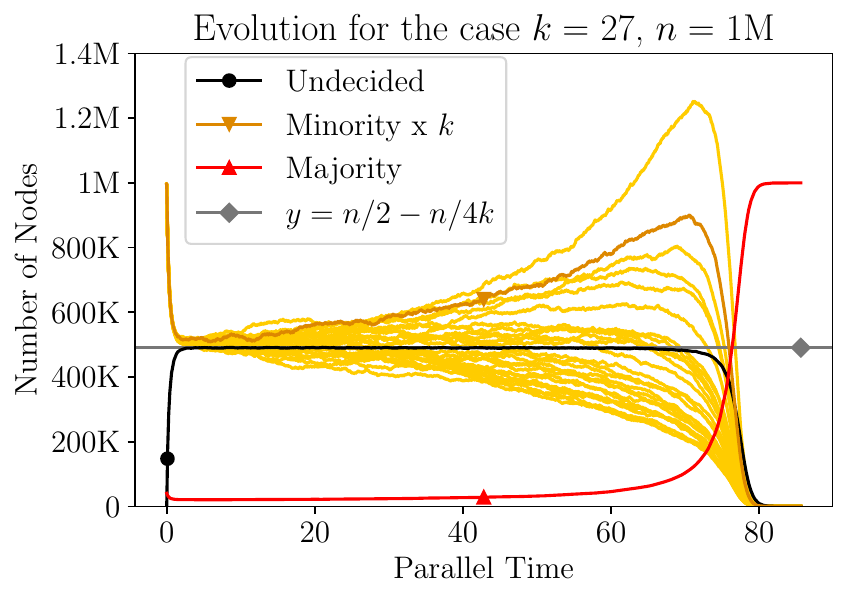}
    \includegraphics[width=0.49\linewidth]{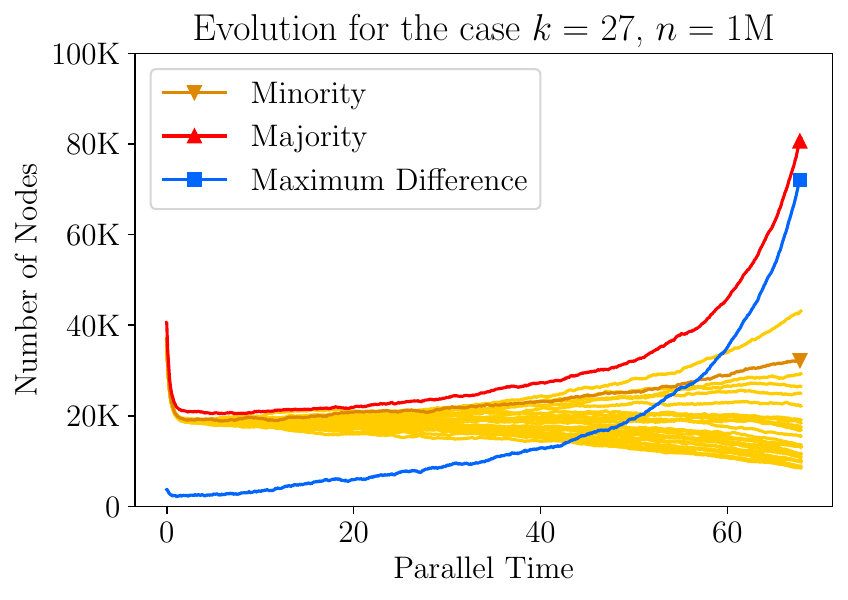}
    \caption{Intuition for stabilization of plurality consensus in undecided state dynamics. On the left figure, we scale up the minority opinions by a $k$ factor for visibility. In this experiment, $n = 1,000,000$ agents and $k = \frac {\sqrt{n}}{\log n \log \log n}$.
In the initial configuration $k-1$ opinions had the same support, while opinion $1$ had an additive advantage of $\sqrt{n \log n}$. Out of all the minority opinions, we plot one with a darker color for visibility.}
\Description[Results of a simulation on 1 million agents]{The picture on the left shows the full experiment, from parallel time 0 to time 85. The picture on the right zooms in on times from 1 to 65 for better visibility. The majority opinion stays close to 0 for parallel time from 0 to 65, then surges to reach 1M around time 85. The number of undecided nodes goes from 0 to roughly n/2 -n/4k in 5 parallel time, stays very clos to that value until time 65, then collapses to 0. the minority opinions all stay close to 0 thoughout, after a steep decline in time 0 to 5, some of them increase slowly, some decrease slowly, but they all collapse after time 65. The maximum difference between the majority opinion and the other opinions seems to grow like a slow exponential.}
    \label{fig:simulation}
\end{figure*}

To illustrate the evolution of opinions and the stabilization behavior, \Cref{fig:simulation} (left) plots a typical simulation run over parallel time. In this example, all minority opinions are initially equally frequent, while the majority one has an initial additional bias of $\sqrt{n \log n}$. The majority opinion is plotted in red. The minority opinions are plotted in yellow, except for a random one which we select as an example and which we highlight in orange; for better visibility we multiplied the number of each minority opinion by $k$. 

We make several observations which highlight parts of the complexity of the problem.
First, note that different minority opinions evolve differently. In particular, not all minority opinions are strictly decreasing over time, but many are actually increasing over a long time period. In the example in \Cref{fig:simulation}, one opinion even surpasses its initial count. 
On the other hand, the majority opinion in this example (which is typical for many runs) remains low for a long time during stabilization, but then increases quickly toward the end. We can also see that the number of undecided opinions stays close to $\frac n 2 - \frac n {4k}$ throughout the execution.    

In \Cref{fig:simulation} (right), we zoom in on the time period for which it takes $x_1$ to double from its initial number of nodes.
We can see that it takes most of the stabilization time to reach this configuration, as the system needs around $70$ parallel time to do so, after which only $20$ more rounds are required to fully stabilize.

These simulations also provide some intuitive motivation for our choices for analyzing the problem theoretically in the next section. We will prove that the number of undecided nodes does not substantially exceed $\frac n 2 - \frac n {4k}$ with high probability, and we only consider the interactions before $x_1$ reaches $2 x_1(0)$, which we will approximate as $2\frac n k$.
One particular ingredient we consider in our analysis is the \emph{maximum} difference between the majority opinion and minority ones (also plotted in \Cref{fig:simulation} (right)): $\max_{j\ge 2}\{x_1-x_j\}$.
The idea is that as long as this difference stays small, the system is very slow to change.

More concretely, for our analysis, we proceed as follows.

We first observe that a precise characterization of the number of undecided nodes $u(t)$ is key to understand the state dynamics as, intuitively, for any opinion $i$, the larger the number of undecided nodes, the more new nodes can the $i$-opinionated nodes ``convert'' to the opinion $i$, hoping to compensate for the $i$-opinionated nodes that meet nodes with other opinions and thus become themselves undecided. In fact, in their analysis of the stabilization time, Amir, Aspnes, Berenbrink, Biermeier, Hahn, Kaaser and Lazarsfeld~\cite{DBLP:conf/podc/AmirABBHKL23} derived an upper bound and a lower bound on the number of undecided nodes throughout the interactions: $\frac n 2 -\frac {x_1} 2 \le u(t) \le \frac n 2$ for any $t$ after the first $n \log n$ interactions.
More precisely, the evolution of the number of nodes in opinion $i$ is driven by the number of undecided nodes.
For each opinion, there exists a value $u_i$ of the number of undecided nodes that acts as a threshold. That is, if the number of undecided nodes $u$ is above $u_i$, then the corresponding opinion $x_i$ increases, whereas if $u$ is below the threshold, then $x_i$ decreases.
The threshold is a decreasing function in the number of nodes in opinion $i$: the larger $x_i$ is, the smaller $u_i$ is.
One can see this in action in \Cref{fig:simulation} on the left.
When $u(t)$ is still very small in the beginning, all of the opinions decrease quickly.
Then, once $u(t)$ settles around $\frac{n}{2} - \frac{n}{4k}$, some opinions start to steadily increase (even minority ones) while others decrease.
One can understand this the following way: in the initial steps, where the number of undecided nodes increase quickly, some opinions, due to randomness, end up having an advantage on the others.
Towards the end, the number of undecided nodes starts dropping, it thus goes below all thresholds but the one of the majority opinion, which means all opinions drop quickly apart from the majority one.

In our case, for the study of the lower bound, we fully control the initial configuration of the opinions, and we are able to give a better upper bound on $u(t)$ over time, which happens to be very close to the threshold of any opinion $i$.
To do so, we use drift analysis (for a nice introduction, we refer to Lengler~\cite{driftanalysis}).
More specifically, we use a result by Oliveto and Witt~\cite{DBLP:journals/tcs/OlivetoW15}.
Here is a quick intuition on drift analysis: assuming we know the current configuration of the system, we compute the expectation of the change in the number of agents in each state after the next interaction. 
The idea is, if a number changes in expectation by a value $\alpha$ at each interaction, and our goal is to show that the actual number does not wander off by more than $\beta$ from its original value, this takes at least $\Omega(\beta/\alpha)$ many interactions. 
A few more hypotheses are needed to ensure that this result holds with high probability, and to ensure that the probabilities are well-behaved, but this is the main idea.

It turns out that $u(t)$ settles around $\frac n 2 - \frac n {4k}$, and to show that $u(t)$ never substantially exceeds this value, we prove that if it slightly exceeds this value, then at each interaction, in expectation, $u(t)$ decreases by at least $\sqrt{\log n / n}$.
Drift analysis tools then allow us to ensure that $u(t)$ will drift no more than (roughly) $\theta (\sqrt{n \log n})$ away from that position.

The second step we make after giving an upper bound on $u(t)$, is to show that in $\theta(kn)$ interactions, no opinion can go from $\frac {3n}{2k}$ nodes to $\frac{2n}{k}$ agents.
Here, classic drift analysis fails to help, and that for one main reason: drift analysis essentially looks at how much the expectation increases or decreases over time, but sometimes, the expectation gets beaten by the variance of the process, and drift analysis fails to capture the lack of concentration of the process.
To understand that, consider a random walk starting at $0$ and which at each step either increases or decreases by $1$, each with probability $1 / 2$.
The expectation of its increase is $0$ so in expectation, after $m$ steps, it is still at $0$.
However, we know that with high probability it will have reached $\theta(\sqrt m)$ at some point during those $m$ steps, as one can see the position after $m$ steps as the result of a binomial distribution with parameters $(1 / 2, m)$, whose standard deviation is of the order of $\sqrt m$.

In our case, we avoid this problem by looking more carefully at the probabilities involved.
By considering the evolution of $x_i$, while its expectation increases slowly over time, it is not overtaken by the variance.
To see that, think again of our example of our $\pm 1$ random walk: Imagine that now, the random walk increases with probability $p / 2$, decreases with probability $p / 2$, and stays put with probability $1 -p$.
Look now at what happens after $m$ steps.
The walk actually moved for $pm$ out of those steps, and thus the standard deviation is now around $\sqrt{pm}$.
If $p = o(1)$, this has a significant impact.

This is exactly what happens in our case, where for $x_i$ to change, in an interaction, a node with opinion $i$ must have been chosen.
If at most $2\frac n k$ nodes with this opinion exist, then the probability of this event happening is of the order of $ 1/ k = o(1)$.

While we could have chosen to stop here for our analysis, and thus give a $\Omega(k)$ lower bound, we go one step further: there is a weakness in this analysis, i.e., we overestimate the number of nodes in opinion $i$, by stating that the initial count is at most $\frac {3n}{2k}$, while in practice, most opinions stay close to $\frac n {2k}$ for most of the process.
This (high) estimate of $2\frac n k$ then gives a too low threshold on $u$ for $x_i$ to increase, and thus the upper bound on the rate at which $x_i$ increases is not tight enough.
To improve the result, we would need to give a better estimate on the initial $x_i$, and show that it is close to $\frac n {2k}$.
This is not straightforward, as the initial count is close to $\frac n k$, and we would thus need to analyze the initial phase where $u$ increases sharply while all the $x_i$ decrease.

We find a workaround for this difficulty: while it might take effort to accurately approximate $x_i$ after the first quick-changing phase, a value that is easier to estimate is $\Delta_{ij}=x_i-x_j$.
Here, instead of requiring a good estimate on $x_i$ and $x_j$ to analyze the evolution of $\Delta_{ij}$, it turns out that we only need their order of magnitude, as well as a good estimate of $\Delta_{ij}$, which we have.
Therefore, knowing that $x_i \le 2\frac n k$ for $\theta(kn)$ iterations allows us to show that the maximum $\Delta_{ij}$ needs at least $\theta(kn)$ iterations to double.
However, if all $\Delta_{ij}$ are small enough (if they are all $o(\frac n k)$), then it means that no opinion drifts far from the others, and thus $x_i \le \frac {3n}{2k}$ after those $\theta(kn)$ interactions.

Hence, our high level argument works as follows: First give a good estimate on $u$, then 
on the number of nodes in each opinion, and on the maximum difference between opinions:
As long as $x_i \le \frac {3n}{2k}$ at a time $t$ and $\Delta_{ij} = o(\frac n k)$ for all $i, j$, then the order of $x_i$ does not change in the next $\theta (kn)$ interactions, which in turn is used to show that $\Delta_{ij}$ does not double during those same interactions.
This in turn shows that $x_i \le \frac {3n}{2k}$, which allows us to restart another iteration of the induction.
The induction holds for $\log\left(\frac{\sqrt n}{k \log n}\right)$ iterations, which then gives the lower bound.

\section{Lower bound}\label{sec:lowerbound}

 In this section, we will give a lower bound on the number of rounds needed for the protocol to stabilize.
For that, we analyze the situation where all minority opinions start with the same number of nodes, that is, $x_i(0) = x_j(0)$ for any $i,j \in [2,k]$, while the majority opinion starts with an initial gap of $x_1(0)-x_2(0)=O\left((\sqrt{n}/(k \log n))^{1/4} \sqrt{n \log n}\right)$. We also require that $k = o\left( \frac {\sqrt n}{\log n}\right)$.
Since a $\Omega(\log n)$ lower bound for the problem is trivial (in $o(\log n)$ parallel time, w.h.p. there are nodes that have not interacted at all) which implies a $\Omega(k\log n)$ lower bound for $k = O(1)$, we focus on the case $k = \omega(1)$.

We first begin by giving an upper bound on the number of undecided nodes.

\begin{lemma}\label{lem:stableu}
    For any $\tau \leq n^4$, and any initial configuration, it holds with probability at least $1-n^{-4}$ that $u(\tau) \le \frac n 2 - \frac n {4k} + \frac {10n} {(k-1)^2} + (20\cdot 132 +1)\sqrt{n \log n}$. 
\end{lemma}
\begin{proof}
To prove the lemma, we will use \Cref{thm:driftlowbound}.
We model $u(t)$ to be a random walk over integers, and first compute the expectation of $u(t+1)$ conditioned on the configuration of the system at time $t$, which is denoted by $\x = (x_1, \dots, x_k, u)$ (for ease of notation, we write $x_i$ instead of $x_i(t)$ and $u$ instead of $u(t)$).
Clearly, the number of undecided nodes can either decrease by $1$, if a decided node interacts with an undecided node, or it can increase by two, if two nodes of different opinions interact.
The probability to decrease by one is $2\frac u n \cdot \frac {n-u} {n-1}$, while the probability 
to increase by two is $\frac{\sum_{i\in [k]}x_i}{n}  \cdot \frac{\sum_{j \neq i}x_j} {n-1}$.
Note that $\frac 1 {n-1} = \frac 1 n +O(\frac 1 {n^2})$.
We therefore have:
        \begin{align*}
    \expect[&u(t+1)\big|\x(t)=\x] \\&= u-2\frac u n \cdot \frac {n-u} n + 2 \frac{\sum_{i\in [k]}x_i}{n}\frac{\sum_{j \neq i}x_j} n +O\left(\frac 1 n \right)
    \\
    &= u - 2 \frac u n +2 \frac {u^2} {n^2} +2 \frac{\sum_{i \in [k]}x_i (n-u-x_i)}{n^2}+O\left(\frac 1 n \right)
    \\
    \intertext{For the second equality we used $\sum_{j \neq i}x_j = n-u-x_i$. Splitting $x_i(n-u-x_i)$ into $x_i(n-u)$ and  $-x_i^2$, and writing $\sum_{i \in [k]}x_i = n-u$ we obtain}
    &= u-2\frac u n +2\frac {u^2} {n^2} +2 \frac {n^2-2un+u^2}{n^2} -2\frac {\sum_{i \in [k]}x_i^2}{n^2}+O\left(\frac 1 n \right)\\
    &= u -6 \frac u n +4 \frac{u^2}{n^2} +2 -2\frac{\sum_{i \in [k]}x_i^2}{n^2}+O\left(\frac 1 n \right)
    \end{align*}
     Clearly, if we fix $n-u$, the sum $\sum_{i \in [k]}x_i^2$ is minimized if all $x_i$ are equal. Then we have
    \begin{align*}
    \expect[&u(t+1)\big|\x(t)=\x]\\
    &\leq u-6\frac u n + 4 \frac{u^2}{n^2} +2 -2 \frac{(n-u)^2}{k n^2} +O\left( \frac 1 n \right) \\
    &= u-6\frac u n + 4 \frac{u^2}{n^2} +2 - \frac 2 k + \frac{4u}{nk} -\frac{2u^2}{n^2 k} +O\left( \frac 1 n \right) 
    \end{align*}
    Let us assume that $u = \frac n 2 - \frac n {4k}+\frac {10n} {(k-1)^2} + c\sqrt{n\log n}$ for some constant $c$. Then, 
    \begin{align*}
    \expect[u(t+1)\big|&\x(t)=\x]\\
    &\leq
    u - 6\frac{\frac n 2 - \frac n {4k}+\frac {10n} {(k-1)^2} + c\sqrt{n\log n}}{n} \\ & \quad  + 4 \frac{(\frac n 2 - \frac n {4k}+\frac {10n} {(k-1)^2} + c\sqrt{n\log n})^2}{n^2} + 2 - \frac 2 k \\
    &\quad  + \frac{4(\frac n 2 - \frac n {4k}+\frac {10n} {(k-1)^2} + c\sqrt{n\log n})}{nk}\\ &\quad - \frac{2(\frac n 2 - \frac n {4k}+\frac {10n} {(k-1)^2} + c\sqrt{n\log n})^2}{n^2 k} +O\left( \frac 1 n \right) \\    %
    &\leq u - 3 + \frac{3}{2k} - \frac{60}{(k-1)^2} - 6 c \sqrt{\frac{\log n}{n}} \\ &\quad+1  +\frac{1}{4k^2} - \frac{1}{k} + \frac{40}{(k-1)^2} \\
    &\quad + 4c \sqrt{\frac{\log n}{n}} + 2 - \frac 2 k + \frac 2 k - \frac{1}{k^2}  - \frac{1}{2k} 
    + \frac{1}{2k^2} \\ &\quad+ O\left(\frac{1}{k^3}\right)+o\left(\sqrt{\frac{\log n}{n}}\right) \\
    &\leq u - c \sqrt{\frac{\log n}{n}}
    \end{align*}
    for any $c \geq 1$ if $k$ is large enough. Define $\tilde{u} = \frac n 2 - \frac n {4k}+\frac {10n} {(k-1)^2}$. %

We can now apply \Cref{thm:driftlowbound} with $X_t=-u(t), X_0=0, a=-\tilde{u}-\sqrt{n\log n}-20\cdot 132 \sqrt{n\log n}, b=-\tilde{u}-\sqrt{n\log n} , \ell = 20\cdot 132\sqrt{n\log n}, \epsilon = \sqrt{\frac {\log n}{n}}$ and $r=\sqrt{5}$. %
Then, we obtain that $T^*$, the first time so that $u(T^*) \ge \frac n 2 - \frac n {4k} + \frac {10n} {(k-1)^2} +\sqrt{n\log n}+ 20\cdot 132\sqrt{n\log n}$, satisfies:

$$
\Proba[T^* \le \exp (4\log n)] \le O(\exp (-4\log n))
$$
This implies that $\Proba[T^* \le n^4] \le O(n^{-4})$.
Therefore, with high probability $u(t)$ is less than $\frac n 2 - \frac n {4k} + \frac {10} {(k-1)^2} + \sqrt{n \log n} +20\cdot 132 \sqrt {n \log n}$ 
for all $0\le t \le n^4$.
\end{proof}

We will need the following lemma to continue our analysis.
It is a simplified version of Theorem~20 in~\cite{driftanalysis}, and provides the tool we need for Lemmas \ref{lem:xone} and \ref{lem:deltas}.

\begin{lemma}\label{thm:randomwalk}
    Let $Y(t)$, $t \geq 0$, be a random walk defined over the set of integers as follows: $Y(0) = 0$, $Y(t+1) = Y(t)$ with probability $1-p(t)$, $Y(t+1) = Y(t)+1$ with probability $\frac {p(t) + q(t)} {2}$, and $Y(t+1) = Y(t)-1$ with probability $\frac {p(t) - q(t)} {2}$.
    Assume there are values $p>0$ and $q>0$ such that $0 \le p(t) \le p$ and $-p(t) \le q(t) \le q $.
    For any $T$ such that $T \ge 32\left( \frac {p - q^2} {2q} + \frac 2 3 \right) \log n$,
    with probability at least $1-n^{-2}$, $Y(t)<T$ for every $t < \min \{\frac T {2q}, n^2\}$ steps.
\end{lemma}
\begin{proof}
To prove this lemma, we will need Bernstein's inequality (cf. \Cref{thm:bernstein}). We first introduce the random variables $\tilde Y(t)$, coupled to $Y(t)$ as follows:
$\tilde Y(t)$ is a random walk such that with probability $1-p(t)$, $\tilde Y(t+1) = \tilde Y(t)$, with probability $\frac {p(t) + q} {2}$, $\tilde Y(t+1) = \tilde Y(t)+1$, and with probability $\frac {p(t) - q} {2}$, $\tilde Y(t+1) = \tilde Y(t)-1$.
Also, with probability $1$, $Y(t+1) = Y(t)$ iif $\tilde Y(t+1) = \tilde Y(t)$ and if $Y(t+1) \ge Y(t)$, then $\tilde Y(t+1) \ge \tilde Y(t)$.

This can easily be done by sampling a random number $r(t) \in (0,1)$ uniformly at random for every $t$.
If $r(t) \le 1-p(t)$, we set $Y(t+1) = Y(t)$ and $\tilde Y(t+1) = \tilde Y(t)$.
If $1-p(t)\le r(t) \le 1-p(t) + \frac{p(t) +q(t)}{2}$, we set $Y(t+1) = Y(t)+1$ and $\tilde Y(t+1) = \tilde Y(t)+1$.
If $1-p(t)+\frac{p(t) +q(t)}{2}\le r(t) \le 1-p(t) + \frac{p(t) +q}{2}$, we set $Y(t+1) = Y(t)-1$ and $\tilde Y(t+1) = \tilde Y(t)+1$.
Else, we set $Y(t+1) = Y(t)-1$ and $\tilde Y(t+1) = \tilde Y(t)-1$.

Essentially, if $Y(t)$ increases, so does $\tilde Y(t)$.
If $Y(t)$ stays put, so does $\tilde Y(t)$.
If $Y(t)$ decreases, $\tilde Y(t)$ might either decrease or increase.
This coupling ensures that $\tilde Y(t) \ge Y(t)$ for all $t$.
To prove the theorem, it thus suffices to prove the theorem for $\tilde Y(t)$.

We will now use Bernstein's inequality with $X_i = \tilde Y({i+1}) - \tilde Y({i}) - q$.
We use $M = 2$ and $\expect [X_i^2] = (1-p(t)) q^2 + \left(\frac {p(t) + q}{2}\right) (1-q)^2 + \left(\frac {p(t) - q}{2}\right) (-1-q)^2 = p(t)-q^2 \le p-q^2$.
Hence for any $N$, $\sum_{i\in[N]} \expect [X_i^2] \le N (p-q^2)$.
With $N \le \frac T {2q}:$
\begin{align*}
    \Proba\left(\tilde Y(N) \ge T \right) &= \Proba\left(\sum_{t=1}^N X_i \ge T -qN\right) 
    \le \Proba\left(\sum_{t=1}^N X_i \ge \frac T 2\right) 
    \\ &\le \exp\left(-\frac {\frac 1 8  T^2}{ \sum_{i\in[N]} \expect [X_i^2]+\frac {2T} {3}} \right)\\
    &\le \exp\left(-\frac {\frac 1 8  T^2}{ N(p-q^2)+\frac {2T} {3}} \right) 
    \le \exp\left(-\frac {\frac 1 8  T}{ \frac {p-q^2} {2q}+\frac {2} {3}} \right) \\ &\le n^{-4}
\end{align*}

Using a union bound over the first $\min \{\frac{T}{2q}, n^2\}$ steps, we get that the probability that $T$ is reached within the first $\min \{\frac{T}{2q}, n^2\}$ steps is at most $n^{-2}$.
\end{proof}

We now prove that w.h.p. 
any given opinion can not increase too much in $\theta(kn)$ interactions, if it starts with fewer than $\frac{3n}{2k}$ nodes.

\begin{lemma}\label{lem:xone}
Let $i \in \{1, \dots, k\}$ be an arbitrary but fixed opinion. 
Let $\tau_{3n/2k} \le n^2$ be a time where
$x_i(\tau_{3n/2k}) \le \frac{3n}{2k}$ and $\tau_{2n/k}$ the random variable denoting the time at which $x_i(\tau_{2n/k}) = \frac{2n}{k}$. 
Then, with probability 
at least $1-O(n^{-2})$ we have $\tau = \tau_{2n/k} - \tau_{3n/2k} \geq \frac {kn}{25}$.
\end{lemma}

\begin{proof}
    Let us consider the evolution of $x_i(t)$ for some $t \in \{t_0, \dots, t_0+ kn/25\}$, where $t_0= \tau_{3n/2k}$. For the analysis, we create a new random process $\y$ similar to $\x$ as follows: $\y(0) =\x(0)$, and for every $t$, $\y(t+1) = \x(t+1)$ if $\y(t) = \x(t)$ and $u(t+1) \le \tilde u + (20\cdot 132+1)\sqrt{n \log n}$. 
    If $u(t+1)$ does not satisfy the condition, we halt $\y$ (we do not define $\y(t+1)$) and say that $\y$ fails at time $t+1$.
    By \Cref{lem:stableu}, $\y$ fails with probability $O(n^{-4})$ in the first $n^4$ rounds.
    The modification described above enforces that in $\y$ the number of undecided nodes never surpasses  $\frac n 2 - \frac n {4k} + \frac {10n} {(k-1)^2} + (20\cdot 132+1)\sqrt{n\log n}$. 
    Clearly, if during the execution of $\x$ the number of undecided nodes does not reach $\frac n 2 - \frac n {4k} + \frac {10n} {(k-1)^2} +\sqrt{n\log n}+ 20\cdot 132\sqrt{n\log n}$, which happens with probability at least $1-n^{-4}$ in the first $n^4$ interactions, then $\x$ and $\y$ behave identically. 

    Let us now consider the evolution of $x_i(t)$ in $\y$.
     $x_i(t)$ increases by one if a node of opinion $i$ meets with an undecided node, which happens with probability 
    $$\Proba(+1) := \Proba(x_i(t+1) - x_i(t) = 1 \big| \y(t) = \y) = 2\frac {x_i}{n} \frac {u} {n-1}$$
    Similarly, the probability that it decreases by one is 
    $$\Proba(-1) := \Proba(x_i(t+1) - x_i(t) \\= -1 \big|\y(t) = \y) = 2\frac {x_i}{n} \frac {n-u-x_i} {n-1} $$
    The goal is to use \Cref{thm:randomwalk}.
    We thus need to compute $p(t) = \Proba(+1) +\Proba(-1)$ and $q(t) = \Proba(+1)-\Proba(-1)$.
    
    As long as $x_i \le 2 \frac n k$, we have that $\Proba(+1) + \Proba(-1) = 2\frac{x_i}{n} \frac {n-x_i}{n-1} = 2\frac{x_i}{n}(1+o(1)) \le \frac 5 k$. The difference between the two is

    \begin{align*}
        \Proba(+1) &- \Proba(-1) \\
        &= 2\frac {x_i} {n-1} \left(\frac u n - \frac {n-u-x_i} n\right) = 2\frac {x_i}{n(n-1)} (2u-n+x_i)
        \\&\le 2\frac {x_i}{n^2} \left(n-\frac {n}{2k} + O\left(\frac n {k \log n}\right)-n + \frac {2 n}{k}\right)\\
        &= \frac {x_i} {n^2}(3+O(\frac 1 {\log n})) \frac {n}{k} = 3\frac {x_i} n \frac {1+o(1)}{k} \le \frac 6 {k^2} (1+o(1))
    \end{align*}  

We apply \Cref{thm:randomwalk}, with $p = \frac 5 k, q= \frac {6.25} {k^2}$, $T= \frac {n} {2k}$.
    It now suffices to check the conditions of \Cref{thm:randomwalk}.
$$
\frac{p-q^2}{2q} \le  \frac{\frac 5 k}{2\frac {6.25} {k^2}} = O(k) = o(k \log n)
$$
and
$$
T=\frac{n}{2k} \ge \frac{k^2\log^2n}{2k} = \omega(k\log^2n)
$$
Therefore, we have that $T \ge 8(\frac{p-q^2}{2q}+\frac 2 3) \log n$. Thus, $x_i(t_0+\tau) < \frac{2n}{k}$ for any 
$\tau \leq \frac {kn}{25}$ with probability at least $1-n^{-2}$ in $\y$. As $\x$ and $\y$ behave identically with probability $1-n^{-4}$ in the first $n^4$ interactions, the lemma follows.
\end{proof}

Now that we have a good upper bound on $u$ and a good estimate on the order of magnitude of all of the $x_i$, we can show that it takes $\theta(kn)$ interactions for the difference between two opinions to double.

\begin{lemma}\label{lem:deltas}
    Assume that at some time step $t_0  \leq nk \log^2 n $ the difference between any two opinions is at most $\frac \alpha 2 = \omega (\sqrt{n \log n})$, where
    $\alpha = o\left(\frac n k\right)$. Furthermore, assume that $x_i(t_0) ) \le \frac{3n}{2k}$ for every $i \in [k]$.
    Then, with probability at least $1-O(k^2/n^{2})$, the difference between any two opinions does not exceed $\alpha$ before interaction $t_0 +\tau$, where $\tau = \frac 1 {24} kn$.
\end{lemma}

\begin{proof}
Consider two arbitrary but fixed opinions $i$ and $j$. 
    Similarly to the proof of Lemma \ref{lem:xone} we analyze the evolution of $\Delta_{ij}(t) =x_i(t) - x_j(t)$ for a time step $t 
    \in \{t_0, \dots , t_0+\tau\}$.
    For the analysis, we create a new random process $\y$ similar to $\x$ as follows: $\y(0) =\x(0)$, and for every $t$, $\y(t+1) = \x(t+1)$ if $\y(t) = \x(t)$ and $u(t+1) \le \tilde u + (20\cdot 132+1)\sqrt{n \log n}$ and $x_i(t+1) \le \frac {2n} k$ for every $i$. 
    If either $u(t+1)$ or one of the $x_i$ does not satisfy the condition, we halt $\y$ (we do not define $\y(t+1)$) and say that $\y$ fails at time $t+1$.
    
     Clearly, $\x$ and $\y$ behave identically with probability $1-O(kn^{-2})$ according to Lemmas  \ref{lem:stableu} and \ref{lem:xone} (by union bound over the $k$ different opinions).

    We consider now the evolution of $\Delta_{ij}(t)$ in $\y$. For ease of notation, we set $x_i =x_i(t)$, $x_i =x_i(t)$, and $u=u(t)$. $\Delta_{ij}(t)$ increases by one if $x_i$ increases by one, which happens when a node with opinion $i$ interacts with an undecided node, hence with probability $2\frac {x_i} n \frac u {n-1}$.
    It also increases by 1 if $x_j$ decreases by one, but not $x_i$, which happens when a node with opinion $j$ interacts with a node of opinion in $[k] \setminus \{i, j\}$, hence with probability $2\frac {x_i} n \frac {n-u-x_i-x_j} {n-1}$.
    Therefore, the probability that $\Delta_{ij}(t+1) > \Delta_{ij}(t)$ is:
    \begin{multline*}
    \Proba(+1) := \Proba(\Delta_{ij}(t+1) - \Delta_{ij}(t) = 1 \big| \y(t)=\y) \\= 2\frac {x_i}{n} \frac u {n-1} + 2\frac {x_j}{n}\frac {n-u-x_i-x_j} {n-1}
    \end{multline*}
    Similarly, the probability that it decreases by one is 
    \begin{multline*}\Proba(-1) := \Proba(\Delta_{ij}(t+1) - \Delta_{ij}(t) = -1 \big| \y(t)=\y) \\= 2\frac {x_j}{n} \frac u {n-1} + 2\frac {x_i}{n}\frac {n-u-x_i-x_j} {n-1}\end{multline*}
    The goal is to use \Cref{thm:randomwalk}.
    We thus need to compute $p(t) = \Proba(+1) +\Proba(-1)$ and $q(t) = \Proba(+1)-\Proba(-1)$.
    Expanding the terms, we get that the sum of $\Proba(+1)$ and $\Proba(-1)$ is:

    \begin{align*}
    \Proba(+1) +\Proba(-1)&= 2\frac{(x_i+x_j) u}{n(n-1)} + 2\frac {(x_i+x_j)}{n(n-1)}(n-u-x_i-x_j)\\
    &\le  2\frac{4 \frac n k}{n(n-1)}(n-x_i-x_j) = 8\frac{1}{k}(1+o(1)) 
    \end{align*}
    The difference between the two is:

    \begin{align*}
        \Proba(+1) &- \Proba(-1)\\ &= \frac 2 {n(n-1)}((ux_i + nx_j - ux_j -x_ix_j -x_j^2) \\&\quad +(-ux_j - nx_i +ux_i +x_i^2 +x_ix_j)) \\
        &= \frac 2 {n(n-1)} (2u(x_i -x_j) -n(x_i - x_j)  \\ &\quad+(x_i - x_j) +(x_i-x_j)(x_i+x_j)) \\
        &= \frac 2 {n(n-1)}(x_i-x_j)(2u-n+x_i+x_j)\\
        &\le \frac{2}{n(n-1)}(x_i-x_j)(n-\frac n {2k} + o (\frac n {k}) -n+2\frac n {k})\\
        &=  \frac 3 {n-1} \Delta_{ij} (\frac 1 k +o(\frac 1 {k}))
    \end{align*}  
    If the difference at a time $t$ is $x_i(t)-x_j(t) \le \alpha$, then  $\Proba(+1) - \Proba(-1) \le \frac {3\alpha} {nk}(1+o(1)) $.
    
    We apply \Cref{thm:randomwalk}, with $p=  \frac 9 k$, $q= \frac {6\alpha} {nk}$, and $T = \frac \alpha 2$, $Y(t) = \Delta_{ij}(t) - \frac \alpha 2$.
    Indeed, we can see the problem as a random walk that starts at $Y(0) = 0$ (which corresponds to $\Delta_{ij} = \frac \alpha 2$) and whose target is $Y(\tau)=\frac \alpha 2$ (which corresponds to $\Delta_{ij} = \alpha$).
    As long as the target is not reached, we have that $\Delta_{ij} \le \alpha$, and thus, $q(t) \le \frac{6\alpha}{nk}$.
    It now suffices to check the conditions of \Cref{thm:randomwalk}.

    We have $T =  \frac \alpha 2 = \omega(\sqrt{n \log n})$, $\frac 2 3 \log n = o(\sqrt {n \log n})$ and $\frac{p-q^2}{q} \log n \le \frac p q \log n = O(\frac n \alpha \log n) = o(\frac n {\sqrt{n \log n} }\log n) = o(\sqrt{n \log n})$, and thus $32\log n \left(\frac{p-q^2}{2q}+\frac 2 3\right) = o(\sqrt{n \log n})\le T = \omega (\sqrt{n \log n})$.
    We thus know that with probability at least $1-n^{-2}$ the difference will not exceed $\alpha$ before $\frac T {2q} = \frac 1 {24} kn$. 

    Taking a union bound over all $k(k-1)/2$ pairs of opinions, we know that the probability that one difference exceeds $\alpha$ at time $t_0 + \tau$ is less than $k^2/n^{-2}$. Since $\x$ and $\y$ as described at the beginning of this proof behave identically with probability $1-O(kn^{-2})$, we obtain the lemma. The lemma also implies that with probability $1-O(k^2/n^{2})$ the support of each opinion is less than $3n/2k$ at time $t_0 + \tau$.   
\end{proof}

We are now ready to prove the main theorem.

\begin{theorem}
For any $k = \omega(1), k= o(\frac {\sqrt n} {\log n})$, any initial configurations where the maximum difference between two opinions is $\max_{i,j \in [k]} \{x_i(0)-x_j(0)\}=O\left((\sqrt{n}/(k \log n))^{1/4} \sqrt{n \log n}\right) $ under the Undecided State Dynamics for Plurality Consensus does not stabilize in $\frac k {25} \log \frac {\sqrt n} {k \log n}$ parallel time, with high probability.    \end{theorem}

\begin{proof}
Define $f$: $f(n)=  \left(\frac{\sqrt n}{k \log n }\right)^{\frac 1 4}$. 

We bunch together the interactions by groups of $\tau =\frac {kn}{25}$ interactions, and use induction on $\ell$: For $\ell \le \log \left( \frac {\frac {n^{\frac 3 4}} {k^{\frac 1 2}}}{\sqrt{n \log n} f(n)}\right)$, after the $\ell$-th group of interactions, we have that $x_i(\ell\tau) \le \frac {3n}{2k}$ for every $i$, and that $\Delta(\ell\tau) := \max_{ij}\Delta_{ij}(\ell\tau) \le 2^\ell\beta$ with probability $1-O(\ell k^2/n^2)$, where the initial bias starts at $\Delta(0) = \beta = O\left(f(n) \sqrt{n \log n}\right)$.

The base step is trivial.
To show the induction step, we first note that for $\ell \le \log \left( \frac {\frac {n^{\frac 3 4}} {k^{\frac 1 2}}}{\sqrt{n \log n} f(n)}\right)$, we have $2^\ell \beta \le \frac {n^{\frac 3 4}}{k^ {\frac 1 2}}$, which in turn means that, using $k \le \frac {n^{\frac 1 2}}{\log n}$, $2^\ell \beta \le \frac {n^{\frac 3 4}} {k^{\frac 1 2}} \cdot \left(\frac {n^{\frac 1 2}} {k\log n}\right) ^{\frac 1 2} =o(\frac  n k)$.
We then apply \Cref{lem:xone,lem:deltas}: First, since for any $i$, $x_i(\ell\tau) \le \frac{3n}{2k}$, and by our result on $u$ (\Cref{lem:stableu}), we apply \Cref{lem:xone} to have that during the whole interval, with probability at least $1-O(n^{-2})$, $x_i$ does not exceed $2\frac n k$.
By union bound over all values of $i$, it holds that for all $i$, no $x_i$ exceeds $2\frac n k$ during the $\frac {kn}{25}$ interactions with probability at least $1-O(kn^{-2})$.
Then, since $\Delta(\ell\tau) \le 2^\ell\beta$, $x_i(t) \le 2 \frac n k$ for any $t \le (\ell+1) \tau$, and by our result on $u$, we can apply \Cref{lem:deltas}.
This ensures that $\Delta((\ell+1)\tau) \le 2^{\ell+1} \beta$ with probability at least $1-O((\ell+1) k^2/n^2)$.
For $\ell+1 \le \log \left( \frac {\frac {n^{\frac 3 4}} {k^{\frac 1 2}}}{\sqrt{n \log n} f(n)}\right)$, this ensures $\Delta((\ell+1)\tau) \le \frac {n^{\frac 3 4}}{k^ {\frac 1 2}}$, which in turn means that, using $k \le \frac {n^{\frac 1 2}}{\log n}$, $\Delta((\ell+1)\tau) \le \frac {n^{\frac 3 4}} {k^{\frac 1 2}} \cdot \left(\frac {n^{\frac 1 2}} {k\log n}\right) ^{\frac 1 2} =o(\frac  n k)$.
    Thus $x_i((\ell+1)\tau) \le \frac {3n} {2k} $ for every $ i \in [k]$, as otherwise $\sum_{j \in [k]} x_j((\ell+1)\tau) \ge \frac 3 2\frac n k + \sum_{j \neq i} x_j((\ell+1)\tau) \ge \frac 3 2\frac n k + (k-1) \frac n k >n$.
    This proves the double induction.  

    Hence, with high probability, the system does not stabilize before $\frac {kn}{25} \log \left( \frac {\frac {n^{\frac 3 4}} {k^{\frac 1 2}}}{\sqrt{n \log n} f(n)}\right) = \frac {kn}{25} \left(\frac 1 2\log \left( \frac {\sqrt n} {k \log n}\right) - \log f(n) \right)=\theta \left(kn \log \frac {\sqrt n} {k\log n}\right)$ interactions, concluding the proof.
\end{proof}

\section{Conclusion}\label{sec:conclusion}

We presented an almost tight lower bound on the stabilization time of the Undecided State Dynamics for plurality consensus in the population protocol model. While our result settles the question about the stabilization time, there are several interesting avenues for future research. In particular, it would be interesting to explore scenarios where (slightly) more memory is available at the nodes and where synchronization is possible to some extent: at which point can we break the lower bound barrier? Another open question concerns the required initial bias. While it is known from previous work that with an initial bias in the order of $O(\sqrt{n})$, the system can stabilize to a minority opinion with non-negligible probability~\cite{CGGNPS18}, we assumed a slightly higher initial bias of $\Omega (\sqrt{n\log n})$ which ensures stabilization to the majority opinion (cf.~\cite{DBLP:conf/podc/AmirABBHKL23}); it remains to close this small gap.

\begin{acks}
    This project has received funding from the European Research Council (ERC) under the European Union's Horizon 2020 research and innovation programme (MoDynStruct, No. 101019564)  \includegraphics[width=0.9cm]{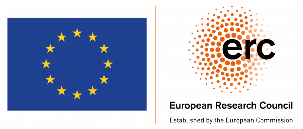} and the Austrian Science Fund (FWF) grant  \href{https://www.doi.org/10.55776/I5862}{DOI 10.55776/I5862}, grant \href{https://www.doi.org/10.55776/I5982}{DOI 10.55776/I5982}, and grant \href{https://www.doi.org/10.55776/P33775}{DOI 10.55776/P33775} with additional funding from the netidee SCIENCE Stiftung, 2020–2024 and the German Research Foundation (DFG), grant 470029389 (FlexNets).
\end{acks}

\appendix
\section{Useful theorems}

\begin{theorem}[Theorem 2 of~\cite{DBLP:journals/tcs/OlivetoW15}]\label{thm:driftlowbound}
    Let $X_t$, $t\ge 0$,be real-valued random variables describing a stochastic process over some state space.
    Suppose there exist an interval $[a,b] \subseteq \R$ and, possibly depending on $\ell := b-a$, a drift bound $\epsilon := \epsilon(\ell) >0$ as well as a scaling factor $r:= r(\ell)$ such that for all $t \ge 0$ the following three conditions hold:
    \begin{itemize}
        \item $\expect[X_{t+1}-X_t | X_0, \dots, X_t; a < X_t <b] \ge \epsilon$.
        \item $\Proba [\card{X_{t+1}-X_t} \ge jr | X_0, \dots , X_t] \le e^{-j}$ for $j \in \N_0$.
        \item $1 \le r^2 \le \frac {\epsilon \ell}{132 \log {\frac r \epsilon}}$.
    \end{itemize}

    Then for the first hitting time $T^\star := \min\{t\ge 0: X_t \le a| X_0, \dots , X_t; X_0 \ge b\}$ it holds that $\Proba\left[ T^* \le \exp\left( \frac{\epsilon \ell}{132r^2}\right)\right] = O\left( \exp\left( -\frac{\epsilon \ell}{132r^2}\right) \right)$
\end{theorem}

\begin{theorem}[Bernstein's Inequality]\label{thm:bernstein}
Let $X_1 , \dots,  X_n$ be independent zero-mean random variables.
Suppose that $\card {X_i} \le M$ almost surely, for all $i$.
Then, for all positive $t$,
$$
    \mathbb {P} \left(\sum _{i=1}^{n}X_{i}\geq t\right)\leq \exp \left(-{\frac {{\tfrac {1}{2}}t^{2}}{\sum _{i=1}^{n}\mathbb {E} \left[X_{i}^{2}\right]+{\tfrac {1}{3}}Mt}}\right)
    $$
\end{theorem}

{\balance
\bibliographystyle{ACM-Reference-Format}
\bibliography{popprotocolbib}
}

\end{document}